\documentclass[conference]{IEEEtran}

\usepackage{amsmath,amssymb,amsthm,mathtools}
\usepackage{bm}
\usepackage{hyperref}

\newtheorem{theorem}{Theorem}
\newtheorem{lemma}{Lemma}
\newtheorem{proposition}{Proposition}

\title{On the Subpacketization Level of the Banawan--Ulukus Multi-Message PIR Scheme}
\author{\IEEEauthorblockN{Anoosheh Heidarzadeh}\IEEEauthorblockA{Department of Electrical and Computer Engineering\\
Santa Clara University, Santa Clara, CA, USA\\
aheidarzadeh@scu.edu}}

\begin{document}
\maketitle
\thispagestyle{plain}

\begin{abstract}
This note analyzes a linear recursion that arises in the computation of the subpacketization level for the multi-message PIR scheme of Banawan and Ulukus.
We derive an explicit representation for the normalized subpacketization level $L$, whose smallest integer multiple yields the subpacketization level of the scheme, in terms of the number of servers $N$, the total number of messages $K$, and the number of demand messages $D$.
The resulting formula shows that
$L$ is a polynomial in $N$ with nonnegative coefficients, and its leading term is $N^{K-D+1}/D$.
\end{abstract}

\section{Setup and Main Results}\label{sec:setup}

Fix integers $K>D>1$ and $N>1$, and let $\{L_1,\dots,L_K\}$ be a sequence satisfying
\begin{align}
&L_K = (N-1)^{K-D}, \label{eq:LK_init}\\
&L_{K-D+1} = L_{K-D+2} = \cdots = L_{K-1} = 0, \label{eq:LK_zeros}\\
&L_j = \frac{1}{N-1}\sum_{i=1}^D \binom{D}{i}\,L_{i+j},
\quad j\in[1:K-D]. \label{eq:LK_rec}
\end{align}
Define
\begin{equation}\label{eq:L_def}
L \triangleq
\frac{N}{D}\sum_{t=1}^{K}\binom{K}{t}L_t
\;-\;
\frac{N}{D}\sum_{t=1}^{K-D}\binom{K-D}{t}L_t.
\end{equation}
The quantity $L$ in~\eqref{eq:L_def} is the normalized subpacketization level of the multi-message PIR scheme in~\cite{BU2018}, meaning that the scheme’s subpacketization level is the smallest integer multiple of $L$~\cite{HWS2025}. Here $N$ denotes the number of servers, $K$ the total number of messages, and $D$ the number of demand messages.

Let
\[
T \triangleq K-D+1,
\quad
S \triangleq \left\lfloor \frac{T(D-1)}{D}\right\rfloor .
\]
For each $n\in[0:T(D-1)]$, define
\begin{equation}\label{eq:cdef}
a_n \triangleq [x^n](1+x+\cdots+x^{D-1})^T,
\end{equation}
where $[x^n]p(x)$ denotes the coefficient of $x^n$ in the polynomial $p(x)$.
Equivalently,
\[
(1+x+\cdots+x^{D-1})^T=\sum_{n=0}^{T(D-1)} a_n x^n.
\]

\begin{theorem}\label{thm:summation}
The quantity $L$ in~\eqref{eq:L_def} can be written as
\begin{equation}\label{eq:Lpoly}
L=\frac{1}{D}\sum_{k=0}^{S} a_{kD}\,N^{T-k}.
\end{equation}
In particular, $L$ is a polynomial in $N$ with nonnegative coefficients, and its leading term is $N^T/D$.
\end{theorem}

\section{Proof of Theorem~\ref{thm:summation}}\label{sec:proof}


Define the reversed sequence
\[
M_t \triangleq L_{K-t},\quad t\in [0:K-1].
\]
Then \eqref{eq:LK_init} and \eqref{eq:LK_zeros} become
\begin{equation}\label{eq:M_init}
M_0=(N-1)^{K-D},\quad M_1=M_2=\cdots=M_{D-1}=0.
\end{equation}

\begin{lemma}\label{lem:reverse_rec}
For each $t\in [D:K-1]$, the recursion \eqref{eq:LK_rec} implies
\begin{equation}\label{eq:M_rec}
N M_t = \sum_{i=0}^{D}\binom{D}{i}M_{t-i}.
\end{equation}
\end{lemma}

\begin{proof}
Fix ${t\in[D:K-1]}$ and set ${j=K-t}$. Then ${j\in [1:K-D]}$, so \eqref{eq:LK_rec} applies:
\[
L_{K-t}=\frac{1}{N-1}\sum_{i=1}^{D}\binom{D}{i}L_{i+K-t}.
\]
Since $L_{K-t}=M_t$ and $L_{i+K-t}=L_{K-(t-i)}=M_{t-i}$, we have
\[
M_t=\frac{1}{N-1}\sum_{i=1}^{D}\binom{D}{i}M_{t-i},
\]
or equivalently $
(N-1)M_t=\sum_{i=1}^{D}\binom{D}{i}M_{t-i}$.
Adding the $i=0$ term $\binom{D}{0}M_t=M_t$ to both sides gives
\[
NM_t=\sum_{i=0}^{D}\binom{D}{i}M_{t-i},
\]
which matches \eqref{eq:M_rec}.
\end{proof}

To solve \eqref{eq:M_rec}, take $M_t=r^t$ with $r\neq 0$. Substitution gives
\begin{align*}
\sum_{i=0}^{D}\binom{D}{i}r^{t-i}=N r^t
&\;\Longleftrightarrow\;
\sum_{i=0}^{D}\binom{D}{i}r^{-i}=N \\
&\;\Longleftrightarrow\;
\left(1+r^{-1}\right)^D=N .
\end{align*}
Let $s>0$ satisfy $s^D=N$ and, for each $m\in [0:D-1]$, define
\[
\omega_m \triangleq e^{\mathrm{i}2\pi m/D},
\quad
u_m\triangleq \omega_m s.
\]

\begin{lemma}\label{lem:rou}
For any integer $P$,
\[
\sum_{m=0}^{D-1}\omega_m^P
=
\begin{cases}
D,& D\mid P,\\
0,& D\nmid P.
\end{cases}
\]
\end{lemma}

\begin{proof}
If $D\mid P$, then $\omega_m^P=(\omega_m^D)^{P/D}=1$ for every ${m\in[0:D-1]}$, and hence
$\sum_{m=0}^{D-1}\omega_m^P=D$.
If $D\nmid P$, then $\omega_1^P\neq 1$ and, since $\omega_m=\omega_1^m$,
\[
\sum_{m=0}^{D-1}\omega_m^P=\sum_{m=0}^{D-1}(\omega_1^P)^m
=\frac{1-(\omega_1^P)^D}{1-\omega_1^P}
=\frac{1-(\omega_1^D)^P}{1-\omega_1^P}
=0,
\]
which proves the lemma.
\end{proof}

\begin{lemma}\label{lem:roots}
The solutions to ${(1+r^{-1})^D=N}$ are
\begin{equation}\label{eq:r_roots}
r_m\triangleq \frac{1}{u_m-1},\quad m\in [0:D-1].
\end{equation}
\end{lemma}

\begin{proof}
Let $r$ be any solution. Since $N=s^D$, the equation $(1+r^{-1})^D=s^D$ implies
$1+r^{-1}=\omega_m s$ for some ${m\in[0:D-1]}$.
Therefore,
\[
r=\frac{1}{\omega_m s-1}=\frac{1}{u_m-1}.
\]
Conversely, each $r_m$ in \eqref{eq:r_roots} satisfies $1+r_m^{-1}=u_m$ and hence
$(1+r_m^{-1})^D=u_m^D=(\omega_m s)^D=s^D=N$.
\end{proof}

Therefore, the general solution to \eqref{eq:M_rec} has the form
\begin{equation}\label{eq:M_general}
M_t=\sum_{m=0}^{D-1}c_m r_m^t,
\end{equation}
where the coefficients $\{c_m\}_{m=0}^{D-1}$ are determined by the initial
conditions \eqref{eq:M_init}.

\begin{lemma}\label{lem:cm}
With coefficients
\begin{equation}\label{eq:cm_choice}
c_m \triangleq \frac{(N-1)^{K-D}}{D}\cdot \frac{(u_m-1)^{D-1}}{u_m^{D-1}},
\quad m\in[0:D-1],
\end{equation}
the sequence $\{M_t\}_{t=0}^{K-1}$ defined by \eqref{eq:M_general} satisfies the
initial conditions \eqref{eq:M_init}.
\end{lemma}

\begin{proof}
Substituting \eqref{eq:cm_choice} into \eqref{eq:M_general} and using $r_m=1/(u_m-1)$,
\begin{align*}
M_t
&=\sum_{m=0}^{D-1}\frac{(N-1)^{K-D}}{D}\cdot \frac{(u_m-1)^{D-1}}{u_m^{D-1}}\cdot \frac{1}{(u_m-1)^t} \\
&=\frac{(N-1)^{K-D}}{D}\sum_{m=0}^{D-1}\frac{(u_m-1)^{D-1-t}}{u_m^{D-1}}.
\end{align*}
For each $t\in[0:D-1]$, expanding $(u_m-1)^{D-1-t}$ and dividing by $u_m^{D-1}$ gives
\begin{align*}
& \frac{(u_m-1)^{D-1-t}}{u_m^{D-1}}\\
& \quad =\sum_{l=0}^{D-1-t}\binom{D-1-t}{l}u_m^{l-(D-1)}(-1)^{D-1-t-l}.
\end{align*}
Therefore,
\begin{align}
& \sum_{m=0}^{D-1}\frac{(u_m-1)^{D-1-t}}{u_m^{D-1}} \nonumber \\
& \quad =
\sum_{l=0}^{D-1-t}\binom{D-1-t}{l}(-1)^{D-1-t-l} \sum_{m=0}^{D-1} u_m^{l-(D-1)}.\label{eq:mysum}
\end{align}
Since ${u_m=\omega_m s}$, for any integer $P$ we have ${\sum_{m=0}^{D-1}u_m^{P}}={s^{P}\sum_{m=0}^{D-1}\omega_m^{P}}$.
By Lemma~\ref{lem:rou}, $\sum_{m=0}^{D-1}\omega_m^P$ equals $D$ if $D\mid P$, and equals zero otherwise.

For ${t=0}$, the exponents ${l-(D-1)}$ range over ${\{-(D-1),\ldots,0\}}$.
The only multiple of $D$ in this set is $0$, which occurs at $l=D-1$ with coefficient $1$.
Thus~\eqref{eq:mysum} equals $D$, and hence ${M_0=(N-1)^{K-D}}$.

For ${1\le t\le D-1}$, the exponents ${l-(D-1)}$ range over ${\{-(D-1),\ldots,-t\}}$, none of which is divisible by $D$.
Therefore \eqref{eq:mysum} equals zero, and consequently $M_t=0$.
\end{proof}

Since $L_t=M_{K-t}$, we can rewrite the sums in \eqref{eq:L_def} as
\begin{align}
\sum_{t=1}^{K}\binom{K}{t}L_t
&=\sum_{t=0}^{K-1}\binom{K}{t}M_t, \label{eq:S1}\\
\sum_{t=1}^{K-D}\binom{K-D}{t}L_t
&=\sum_{t=D}^{K-1}\binom{K-D}{t-D}M_t. \label{eq:S2}
\end{align}

\begin{lemma}\label{lem:bin_ids}
For each ${m\in [0:D-1]}$,
\begin{align}
\sum_{t=0}^{K-1}\binom{K}{t}r_m^t
&=(1+r_m)^K-r_m^K, \label{eq:bin1}\\
\sum_{t=D}^{K-1}\binom{K-D}{t-D}r_m^t
&=r_m^D\left((1+r_m)^{K-D}-r_m^{K-D}\right). \label{eq:bin2}
\end{align}
\end{lemma}

\begin{proof}
By the binomial theorem,
\[
\sum_{t=0}^{K}\binom{K}{t}r_m^{t} = (1+r_m)^K.
\]
Dropping the $t=K$ term yields~\eqref{eq:bin1}.
For \eqref{eq:bin2}, factoring out $r_m^D$ and reindexing gives
\begin{align*}
\sum_{t=D}^{K-1}\binom{K-D}{t-D}r_m^t
&= r_m^D\sum_{t=D}^{K-1}\binom{K-D}{t-D}r_m^{t-D} \\
&= r_m^D\sum_{t=0}^{K-D-1}\binom{K-D}{t}r_m^t \\
&= r_m^D\left((1+r_m)^{K-D}-r_m^{K-D}\right),
\end{align*}
where the last step follows by applying \eqref{eq:bin1} with $K$ replaced by $K-D$.
\end{proof}

\begin{proposition}\label{prop:L_intermediate}
The quantity $L$ in \eqref{eq:L_def} can be written as
\begin{equation}\label{eq:L_intermediate}
L=\frac{N}{D}\sum_{m=0}^{D-1}c_m\,(N-1)\,r_m^D(1+r_m)^{K-D}.
\end{equation}
\end{proposition}

\begin{proof}
Starting from \eqref{eq:L_def} and using \eqref{eq:S1} and \eqref{eq:S2},
\[
L=\frac{N}{D}\left(\sum_{t=0}^{K-1}\binom{K}{t}M_t-\sum_{t=D}^{K-1}\binom{K-D}{t-D}M_t\right).
\]
Substituting $M_t=\sum_{m=0}^{D-1} c_m r_m^t$ and exchanging sums,
\[
L=\frac{N}{D}\sum_{m=0}^{D-1}c_m\left(\sum_{t=0}^{K-1}\binom{K}{t}r_m^t-\sum_{t=D}^{K-1}\binom{K-D}{t-D}r_m^t\right).
\]
Applying Lemma~\ref{lem:bin_ids},
\begin{align*}
&\sum_{t=0}^{K-1}\binom{K}{t}r_m^t-\sum_{t=D}^{K-1}\binom{K-D}{t-D}r_m^t \\
&\quad = \left((1+r_m)^K-r_m^K\right)-r_m^D\left((1+r_m)^{K-D}-r_m^{K-D}\right) \\
&\quad =(1+r_m)^{K-D}\left((1+r_m)^D-r_m^D\right).
\end{align*}
By Lemma~\ref{lem:roots}, $(1+r_m^{-1})^D=N$, equivalently $(1+r_m)^D=N r_m^D$, and thus
$(1+r_m)^D-r_m^D = N r_m^D - r_m^D = (N-1)r_m^D$,
which establishes \eqref{eq:L_intermediate}.
\end{proof}

Using $r_m=1/(u_m-1)$ gives $1+r_m=u_m/(u_m-1)$ and hence
\[
r_m^D(1+r_m)^{K-D}=\frac{u_m^{K-D}}{(u_m-1)^K}.
\]
Substituting \eqref{eq:cm_choice} into \eqref{eq:L_intermediate} yields
\begin{equation}\label{eq:L_root_unity}
L=\frac{N(N-1)^T}{D^2}\sum_{m=0}^{D-1}\frac{u_m^{T-D}}{(u_m-1)^T},
\end{equation} where $T = K-D+1$. 


Rewriting each summand in \eqref{eq:L_root_unity}, we have
\[
\frac{u_m^{T-D}}{(u_m-1)^T}
=
u_m^{-D}\bigl(1-u_m^{-1}\bigr)^{-T}.
\]
Using $u_m=\omega_m s$, we obtain $u_m^{-D}=(\omega_m s)^{-D}=s^{-D}=N^{-1}$ and
$1-u_m^{-1}=1-\omega_m^{-1}s^{-1}$. Substituting into \eqref{eq:L_root_unity} gives
\[
L=\frac{(N-1)^T}{D^2}\sum_{m=0}^{D-1}\bigl(1-\omega_m^{-1}s^{-1}\bigr)^{-T}.
\]
Since $\{\omega_m^{-1}:m\in[0:D-1]\}=\{\omega_m:m\in[0:D-1]\}$, reindexing yields
\begin{equation}\label{eq:L_avg_form}
L=\frac{(N-1)^T}{D^2}\sum_{m=0}^{D-1}\bigl(1-\omega_m s^{-1}\bigr)^{-T}.
\end{equation}

\begin{lemma}\label{lem:ratio}
For each $m\in[0:D-1]$,
\[
\frac{1-N^{-1}}{1-\omega_m s^{-1}}=\sum_{t=0}^{D-1}\omega_m^t s^{-t}.
\]
\end{lemma}

\begin{proof}
Multiplying $(1-\omega_m s^{-1})$ by $\sum_{t=0}^{D-1}\omega_m^t s^{-t}$ gives
\begin{align*}
(1-\omega_m s^{-1})\sum_{t=0}^{D-1}\omega_m^t s^{-t}
&=\sum_{t=0}^{D-1}\omega_m^t s^{-t}
-\sum_{t=0}^{D-1}\omega_m^{t+1}s^{-(t+1)} \\
&=1-\omega_m^D s^{-D}
=1-N^{-1},
\end{align*}
since $\omega_m^D=1$ and $s^D = N$. Dividing both sides by ${(1-\omega_m s^{-1})}$ yields the identity.
\end{proof}

For each $n\in[0:T(D-1)]$, define $a_n$ as in \eqref{eq:cdef}, so that
\begin{equation}\label{eq:poly_expand}
(1+x+\cdots+x^{D-1})^{T}
=
\sum_{n=0}^{T(D-1)} a_n\,x^n .
\end{equation}
In particular, $a_n\ge 0$, since $a_n$ equals the number of $T$-tuples $(n_1,\dots,n_T)\in[0:D-1]^T$
such that $\sum_{l=1}^T n_l=n$.

Fix $m\in[0:D-1]$. By Lemma~\ref{lem:ratio},
\begin{equation}\label{eq:powT}
\frac{(1-N^{-1})^T}{(1-\omega_m s^{-1})^T}
=
\left(\sum_{t=0}^{D-1}\omega_m^t s^{-t}\right)^T.
\end{equation}
Expanding the $T$th power using \eqref{eq:poly_expand} gives
\begin{equation}\label{eq:Texp}
\left(\sum_{t=0}^{D-1}\omega_m^t s^{-t}\right)^T
=
\sum_{n=0}^{T(D-1)} a_n\,\omega_m^{n}\,s^{-n}.
\end{equation}
Substituting \eqref{eq:Texp} into \eqref{eq:powT} gives
\begin{equation}\label{eq:ave_powT}
\frac{(1-N^{-1})^T}{(1-\omega_m s^{-1})^T} = 
\sum_{n=0}^{T(D-1)} a_n\,\omega_m^{n}\,s^{-n}.
\end{equation}
Averaging \eqref{eq:ave_powT} over $m\in[0:D-1]$ and applying Lemma~\ref{lem:rou} yields
\begin{align}
\frac{1}{D}\sum_{m=0}^{D-1}\frac{(1-N^{-1})^T}{(1-\omega_m s^{-1})^T} & =\sum_{n=0}^{T(D-1)} a_n\,s^{-n}\left(\frac{1}{D}\sum_{m=0}^{D-1}\omega_m^{n}\right)\nonumber\\
& =\sum_{k=0}^{S} a_{kD}\,s^{-kD} \nonumber \\ 
& = \sum_{k=0}^{S} a_{kD}\,N^{-k},\label{eq:avg_result}
\end{align} where $S = \lfloor T(D-1)/D\rfloor$. 

Returning to \eqref{eq:L_avg_form} and using $(N-1)^T=N^T(1-N^{-1})^T$, 
\begin{align*}
L
&=\frac{(N-1)^T}{D^2}\sum_{m=0}^{D-1}(1-\omega_m s^{-1})^{-T}\\
&=\frac{N^T}{D^2}\sum_{m=0}^{D-1}\frac{(1-N^{-1})^T}{(1-\omega_m s^{-1})^T}\\
& =\frac{N^T}{D}\left(\frac{1}{D}\sum_{m=0}^{D-1}\frac{(1-N^{-1})^T}{(1-\omega_m s^{-1})^T}\right).
\end{align*}
Substituting \eqref{eq:avg_result} yields
\[
L=\frac{N^T}{D}\sum_{k=0}^{S} a_{kD}\,N^{-k}
=\frac{1}{D}\sum_{k=0}^{S} a_{kD}\,N^{T-k},
\]
which matches \eqref{eq:Lpoly}.

Since \eqref{eq:Lpoly} is a finite sum of nonnegative integer powers of $N$ with coefficients $a_{kD}\ge 0$ for all $k\in [0:S]$, it follows that $L$ is a polynomial in $N$ with nonnegative coefficients.
Moreover, the highest power of $N$ occurs at $k=0$, and
\[
a_0=[x^0](1+x+\cdots+x^{D-1})^T=1,
\]
so the leading term of $L$ is $N^T/D$. This completes the proof of Theorem~\ref{thm:summation}.


\bibliographystyle{IEEEtran}
\bibliography{PIR_PC_Refs}

\end{document}